\documentclass[11pt,letter]{article}

\usepackage{amsmath,amsthm,amsfonts,url}
\usepackage[margin=1.0in]{geometry}
\usepackage{graphicx,color}
\usepackage{graphicx}

\newtheorem{theorem}{Theorem}
\newtheorem{lemma}{Lemma}
\newtheorem{corollary}{Corollary}

\begin{document}

\title{\bf The RNA Newton Polytope and Learnability of Energy Parameters}
\author{Elmirasadat Forouzmand \\
Department of Computer Science \\
Wayne State University\\
Detroit, MI 48202 \\
\url{elmira@wayne.edu}
\and Hamidreza Chitsaz\thanks{to whom correspondence should be addressed}\\
Department of Computer Science \\
Wayne State University\\
Detroit, MI 48202 \\
\url{chitsaz@wayne.edu}}

\maketitle

\begin{abstract}
\noindent{\bf Motivation:} Computational RNA structure prediction is a mature important problem which has received a new wave of attention with the discovery of regulatory non-coding RNAs and the advent of high-throughput transcriptome sequencing. Despite nearly two scores of research on RNA secondary structure and RNA-RNA interaction prediction, the accuracy of the state-of-the-art algorithms are still far from satisfactory. So far, researchers have proposed increasingly complex energy models and improved parameter estimation methods, experimental and/or computational, in anticipation of endowing their methods with enough power to solve the problem. The output has disappointingly been only modest improvements, not matching the expectations. Even recent massively featured machine learning approaches were not able to break the barrier. 

\noindent{\bf Approach:} The first step towards high accuracy structure prediction is to pick an energy model that is inherently capable of predicting each and every one of known structures to date. In this paper, we introduce the notion of \emph{learnability} of the parameters of an energy model as a measure of such an inherent capability. We say that the parameters of an energy model are \emph{learnable} iff there exists at least one set of such parameters that renders \emph{every} known RNA structure to date the minimum free energy structure. We derive a necessary condition for the learnability and give a dynamic programming algorithm to assess it. Our algorithm computes the convex hull of the feature vectors of all feasible structures in the ensemble of a given input sequence. Interestingly, that convex hull coincides with the \emph{Newton polytope} of the partition function as a polynomial in energy parameters. To the best of our knowledge, this is the first approach towards computing the RNA Newton polytope and a systematic assessment of the inherent capabilities of an energy model. 

\noindent{\bf Results:} We demonstrated the application of our theory to a simple energy model consisting of a weighted count of A-U and C-G base pairs. Our results show that this simple energy model satisfies the necessary condition for less than one third of the input unpseudoknotted sequence-structure pairs chosen from the RNA STRAND v2.0 database. For another one third, the necessary condition is barely violated, which suggests that augmenting this simple energy model with more features such as the Turner loops may solve the problem. The necessary condition is severely violated for 8\%, which provides a small set of hard cases that require further investigation.   
\end{abstract}

\section{Introduction}
Computational RNA structure and RNA-RNA interaction prediction have always been important problems, particularly now that RNA has been shown to have key regulatory roles in the cell \cite{Storz02,Bartel04,Hannon02,ZamHal05,WagFla02,Bra02,Gottesman05}. Furthermore, with the advent of synthetic biology at the whole organism level \cite{Gibson10}, high-throughput accurate RNA engineering algorithms are required for both \emph{in vivo} and \emph{in vitro} applications \cite{See05,SeeLuk05,SimDit05,Ven07,Reif08}. Since the dawn of RNA secondary structure prediction nearly two scores ago \cite{Borer74}, the research community has proposed increasingly complex models and algorithms, hoping that refined features together with better methods to estimate their parameters would solve the problem. Early approaches considered mere base pair counting, followed by the Turner thermodynamics model which was a significant leap forward. Recently, massively feature-rich models empowered by parameter estimation algorithms have been proposed, but they provide only modest improvements.

Despite significant progress in the last three decades, made possible by the work of Turner and others \cite{MatTur99} on measuring RNA thermodynamic energy parameters and the work of several groups on novel algorithms \cite{Nus78,WatSmi78,Zuker81,RivEdd99,DirPie03,Mcc90,Chitsaz09,Chitsaz09b,BerTaf06} and machine learning approaches \cite{Do06,Andronescu10,Zakov11}, the RNA structure prediction accuracy has not reached a satisfactory level yet. Why is it so? Up to now, human intuition and computational convenience have lead the way. We believe that human intuition has to be equipped with systematic methods to assess the suitability of a given energy model. Surprisingly, there is not a single method to assess whether the parameters of an energy model are \emph{learnable}. We say that the parameters of an energy model are \emph{learnable} iff there exists at least one set of such parameters that renders \emph{every} known RNA structure to date, determined through X-ray or NMR, the minimum free energy structure. Equivalently, we say that the parameters of an energy model are learnable iff 100\% structure prediction accuracy can be achieved when the training and test sets are identical. The first step towards high accuracy structure prediction is to make sure that the energy model is inherently capable, i.e. its parameters are learnable. In this work, we provide a necessary condition for the learnability and an algorithm to verify it. To the best of our knowledge, this is the first approach towards a systematic assessment of the suitability of an energy model. Note that a successful RNA folding algorithm needs to have the generalization power to predict unseen structures as well. We do not deal with the generalization power in this work and leave it for future work. 

\section{Background}
\subsection{RNA Secondary Structure Models}
An RNA secondary structure model is often a context free grammar together with a scoring function for either the rules, in the case of stochastic context free grammars (SCFG) \cite{Eddy94}, or the alphabet, in the case of thermodynamics models \cite{MatTur99}. Such scoring functions induce scoring on the entire generated language. The word with optimal score then yields a predicted structure for the given sequence. For the sake of brevity, we focus on thermodynamics models in this paper, but it is obvious that our methods apply to other models including SCFG as well. In our context, the scoring function is the thermodynamics free energy. A secondary structure $y$ of a nucleic acid is decomposed into loops; a free energy is associated with every loop in $y$; and the total free energy $G$ for $y$ is the sum of loop free energies \cite{MatTur99}. The same loop decomposition principle applies to interacting nucleic acids  such that the total free energy $G$ is still the sum of the free energies of loops and interaction
components \cite{Chitsaz09}. 

\subsection{Estimation of Energy Parameters}
Existing machine learning algorithms for parameter estimation in RNA  structure prediction can be grouped into two categories: 
\begin{itemize}
\item Likelihood-based methods, where the maximum likelihood (ML) principle is used to estimate the parameters of a probabilistic model, e.g. \cite{Do06}, and
\item Large-margin methods, where the model parameters are estimated to maximize the margin between the score of the true structure and the second best structure. This has been done using an online passive-aggressive training algorithm \cite{Zakov11} and Iterative Constraint Generation (CG) \cite{Andronescu07}.
\end{itemize}
The likelihood-based techniques estimate the best \emph{Gibbs} distribution, which not only assists in predicting the best secondary structure but also is utilized in determining the thermodynamic parameters. The most successful method for learning the thermodynamics of RNA has been the maximum likelihood method, as in CONTRAfold \cite{Do06}, which maximizes the probability of RNA structures $y$ given RNA sequences $x$ for the training set $D$. That is, the conditional log likelihood of the training data (using the Boltzmann distribution) is maximized to estimate the best model parameters $\mathbf{h}^* \in \mathbb{R}^k$: 
\begin{align}
\mathbf{h}^* := \arg\max_{\mathbf{h}} L(D; \mathbf{h}) &= \max_{\mathbf{h}} \sum_{(x, y)\in D} \log p(y | x, \mathbf{h}), \label{equ:ml}\\
p(y| x, \mathbf{h}) &:= \frac{e^{-G(x, y, \mathbf{h})/RT}}{Q(x, \mathbf{h})}, \label{equ:prob}
\end{align}
where $k$ denotes the number of different motifs defined in the energy model, $R$ is the gas constant, $T$ is the absolute temperature, $G(x, y, \mathbf{h})$ is the free energy, and 
\begin{equation}\label{equ:partition}
Q(x, \mathbf{h}) := \sum_{s \in \mathcal{E}(x)} e^{-G(x, s, \mathbf{h})/RT}
\end{equation}
is the \emph{partition function} \cite{Chitsaz09,DirPie03,Mcc90} with $\mathcal{E}(x)$ being the ensemble of possible structures of $x$. The free energy 
\begin{equation}\label{equ:energy}
G(x, s, \mathbf{h}) := \left\langle c(x, s), \mathbf{h} \right\rangle 
\end{equation}
is a linear function of the parameters $\mathbf{h}$ where $c(x, s) \in \mathbb{Z}^k$ is the features vector. 

\section{Learnability}
The question that we ask before parameter estimation is: does there ever exist parameters $\mathbf{h}^\dagger$ such that for every $(x, y) \in D$, $y = \arg\min_s G(x, s, \mathbf{h}^\dagger)$? If the answer to this question is no, then there is no hope that one can  ever achieve 100\% accuracy using the given model. The answer reveals inherent limitations of the model, which can be used to design improved models. We provide a necessary condition for the existence of $\mathbf{h}^\dagger$ and a dynamic programming algorithm to verify it through computing the Newton polytope for every $x$ in $D$. We will define the RNA Newton polytope below. Not only our algorithm provides a binary answer, it also quantifies the distance from the boundary.

\section{Methods}\label{sec:methods}

\subsection{Necessary Condition for Learnability}
Let $(x, y) \in D$ and $\mathbf{h}^\dagger \in \mathbb{R}^k$. Assume $y$ minimizes $G(x, s, \mathbf{h}^\dagger)$ as a function of $s$. In that case 
\begin{eqnarray}
G(x, y, \mathbf{h}^\dagger) \leq G(x, s, \mathbf{h}^\dagger),\; \ \forall\ s \in \mathcal{E}(x).
\end{eqnarray}
Replacing (\ref{equ:energy}) above,
\begin{eqnarray}
& \left\langle c(x, y), \mathbf{h}^\dagger \right\rangle \leq \left\langle c(x, s), \mathbf{h}^\dagger \right\rangle,\; \ \forall\ s \in \mathcal{E}(x) \\
& 0 \leq \left\langle c(x, s)-c(x, y), \mathbf{h}^\dagger \right\rangle,\; \ \forall\ s \in \mathcal{E}(x).\label{equ:optimality}
\end{eqnarray}
Define the \emph{feature ensemble} of sequence $x$ by
\begin{equation}
\mathcal{F}(x) := \left\{ c(x, s)\ |\ s \in \mathcal{E}(x) \right\} \subset \mathbb{Z}^k.
\end{equation}
In that case, (\ref{equ:optimality}) implies that
\begin{equation}\label{equ:hyperplane}
0 \leq \left\langle \mathcal{F}(x) - c(x, y), \mathbf{h}^\dagger \right\rangle.
\end{equation}
We call the convex hull of $\mathcal{F}(x)$ the \emph{Newton polytope} of $x$,
\begin{equation}
\mathcal{N}(x) := \mbox{conv}\left\{\mathcal{F}(x)\right\} \subset \mathbb{R}^k.
\end{equation}
We remind the reader that the convex hull of a set, denoted by `conv' hereby, is the minimal convex set that fully contains the set. 
The reason for naming this polytope the Newton polytope will be made clear below. 
Inequality (\ref{equ:hyperplane}) implies that $c(x, y) \in \partial \mathcal{N}(x)$ is on the boundary of the convex hull of the feature ensemble of $x$ with a support hyperplane normal to $\mathbf{h}^\dagger$. Therefore, we have the following theorem.

\begin{theorem}
Let $(x, y) \in D$ and $0 \not= \mathbf{h}^\dagger \in \mathbb{R}^k$. Assume $y$ minimizes $G(x, s, \mathbf{h}^\dagger)$ as a function of $s$. In that case, 
$c(x, y) \in \partial \mathcal{N}(x)$, i.e. the feature vector of $(x, y)$ is on the boundary of the Newton polytope of $x$. 
\end{theorem}
\begin{proof}
To the contrary, suppose $c(x, y)$ is in the interior of $\mathcal{N}(x)$. Therefore, there is an open ball of radius $\delta > 0$ centered at $c(x, y)$ completely contained in $\mathcal{N}(x)$, i.e.
\begin{equation}
B_\delta(c(x, y)) \subset \mathcal{N}(x).
\end{equation}
Let $$p = c(x, y) - (\delta/2)\frac{\mathbf{h}^\dagger}{\|\mathbf{h}^\dagger\|}.$$ It is clear that $p \in B_\delta(c(x, y)) \subset \mathcal{N}(x)$ since $\|p-c(x, y)\| = \delta/2 < \delta$. Therefore, $p$ can be written as a convex linear combination of the feature vectors in $\mathcal{F}(x) = \{v_1, \ldots, v_N\}$, i.e.
\begin{align}
\exists\; \ \alpha_1, \ldots \alpha_N \geq 0:\; \ \alpha_1 v_1 + \cdots + \alpha_N v_N &= p,\\
\alpha_1 + \cdots + \alpha_N &= 1. 
\end{align}
Note that
\begin{equation}\label{equ:negative}
\left\langle p - c(x, y), \mathbf{h}^\dagger\right\rangle = -(\delta/2) \|\mathbf{h}^\dagger\| < 0.
\end{equation}
Therefore, there is $1 \leq i \leq N$, such that $\left\langle v_i - c(x, y), \mathbf{h}^\dagger\right\rangle < 0$ for otherwise,
\begin{equation}
\left\langle p - c(x, y), \mathbf{h}^\dagger\right\rangle = \sum_{i=1}^N \alpha_i \left\langle v_i - c(x, y), \mathbf{h}^\dagger\right\rangle \geq 0,
\end{equation}
which would be a contradiction with (\ref{equ:negative}). It is now sufficient to note that $v_i \in \mathcal{F}(x)$ and $\left\langle v_i - c(x, y), \mathbf{h}^\dagger\right\rangle < 0$ which is a contradiction with (\ref{equ:hyperplane}). \qed
\end{proof}

\begin{corollary}[Necessary Condition for the Learnability]
For $(x, y) \in D$, a necessary but not sufficient condition for the existence of $\mathbf{h}^\dagger$ such that $y$ minimizes $G(x, s, \mathbf{h}^\dagger)$ as a function of $s$ is that $c(x, y)$ lies on the boundary of $\mathcal{N}(x)$ the Newton polytope of $x$. 
\end{corollary}

\subsection{Relation to the Newton Polytope}
In addition to $D$ the set of experimentally determined structures, we often have a repository of thermodynamic measurements, e.g. melting curves, that can help better estimate the energy parameters. Such measurements often relate to the energy parameters through equations involving the partition function and its derivatives with respect to temperature \cite{Chitsaz09}. We show that with a change of variables, the partition function becomes a polynomial. Therefore, such equations become a system of polynomial equations the solving of which algebraically requires computation of the Newton polytope of each polynomial \cite{Emiris94,Emiris95}. Recall the partition function defined in (\ref{equ:partition}) and energy in (\ref{equ:energy}), and conclude
\begin{equation}\label{equ:partition2}
Q(x, \mathbf{h}) = \sum_{s \in \mathcal{E}(x)} e^{-\left\langle c(x, s), \mathbf{h}\right\rangle /RT}.
\end{equation}
Let $c(x, s) = (c_1(x, s), \ldots, c_k(x, s))$ and $\mathbf{h} = (\mathbf{h}_1, \ldots, \mathbf{h}_k)$. Define new variables
\begin{equation}
Z_i := e^{-\mathbf{h}_i/RT},\; \ 1 \leq i \leq k,
\end{equation}
and replace them in (\ref{equ:partition2}). We obtain the partition function
\begin{equation}
Q(x, Z) = \sum_{s \in \mathcal{E}(x)} Z^{c(x, s)},
\end{equation}
in the form of a polynomial in $\mathbb{R}[Z]$ where
\begin{equation}
Z^{c(x, s)} := \prod_{i=1}^k Z_i^{c_i(x, s)}
\end{equation}
is a monomial as $0 \leq c_i(x, s) \in \mathbb{Z}$. The Newton polytope of $Q$ is defined to be the convex hull of the monomials power vectors, i.e.
\begin{equation}
\mbox{Newton}\left\{Q(x, Z)\right\} := \mbox{conv}\left(\left\{ c(x, s)\ |\ s \in \mathcal{E}(x)\right\} \right) = \mathcal{N}(x).
\end{equation}
That is why we call $\mathcal{N}(x)$ the Newton polytope of $x$.

\subsection{RNA Newton Polytope Algorithm}
We give a dynamic programming algorithm to compute the Newton polytope for a given nucleic acid sequence $x$. Denote the length of $x$ by $L$ and the $i^{th}$ nucleotide in $x$ by $n_i$. Denote the subsequence of $x$ from the $i^{th}$ to the $j^{th}$ nucleotide, inclusive of ends, by $n_i\cdots n_j$. The following lemma allows us to formulate a divide-and-conquer strategy for computing the Newton polytope, which will in turn lead to our dynamic programming algorithm. 
\begin{lemma}
Let $f$ and $g$ be two polynomials in $\mathbb{R}[Z]$. The Newton polytope of the product of $f$ and $g$ is the Minkowski sum of individual Newton polytopes, and the Newton polytope of the sum of $f$ and $g$ is the convex hull of the union of individual Newton polytopes, i.e.
\begin{eqnarray}
& \mbox{Newton}\left(fg\right) = \mbox{Newton}\left(f\right) \oplus \mbox{Newton}\left(g\right), \\
& \mbox{Newton}\left(f+g\right) = \mbox{conv}\left\{\mbox{Newton}\left(f\right) \cup \mbox{Newton}\left(g\right)\right\},
\end{eqnarray}
in which $\oplus$ represents the Minkowski sum of two polytopes \cite{Emiris94}.
\end{lemma}
This lemma allows us to use the same divide-and-conquer strategy that was used for calculating the partition function \cite{Chitsaz09,DirPie03,Mcc90}. We can use the same recursions (grammar) as in the partition function algorithm but with the Minkowski sum $\oplus$ instead of multiplication, convex hull of union instead of summation, and the corresponding feature vector $c$ instead of $e^{-\langle c, \mathbf{h} \rangle/RT}$. Furthermore, since union is invariant with respect to repetition of points, the dynamic programming is allowed to be redundant, or equivalently the grammar is allowed to be ambiguous. Hence, any complete RNA structure or RNA-RNA interaction prediction dynamic programming algorithm can be transformed into a Newton polytope algorithm by replacing the energy with the corresponding feature vector, summation with the Minkowski sum $\oplus$, and minimization with the convex hull of union. 

{\bf As explained above, we transform any complete partition function or structure prediction dynamic programming algorithm, for single RNA, RNA-RNA interaction, or multiple interacting RNAs, into a Newton polytope algorithm. For the sake of illustration, we explicitly spell below only the case of single RNA with separate A-U and C-G base pair counting energy model. All the other cases are quite trivially obtained following the transformations above.} 

In this case, the feature vector $c(x, s)=(c_1(x, s), c_2(x, s))$ is two dimensional: $c_1(x, s)$ is the number of A-U, and $c_2(x, s)$ the number of C-G base pairs in $s$. Our dynamic programming algorithm starts by computing the Newton polytope for all unit length subsequences, followed by all length two subsequences,$\ldots$, up to the Newton polytope for the entire sequence $x$. We denote the Newton polytope of the subsequence $n_i\cdots n_j$ by $\mathcal{N}(i, j)$, i.e.
\begin{equation}
\mathcal{N}(i, j) := \mathcal{N}(n_i\cdots n_j).
\end{equation} 
The following dynamic programming will yield the result
\begin{equation}\label{equ:dp}
\mathcal{N}(i, j) = \mbox{conv}\left[\cup  
\left\{ \begin{array}{l}
    \mathcal{N}(i,\ell) \oplus \mathcal{N}(\ell+1,j),\; \ \ i \leq \ell \leq j-1\\
    \left\{(1, 0) \right\} \oplus \mathcal{N}(i+1,j-1) \quad \text{if $n_in_j$ = AU$|$UA}\\
    \left\{(0, 1) \right\} \oplus \mathcal{N}(i+1,j-1) \quad \text{if $n_in_j$ = CG$|$GC}\\
    \end{array} \right. \right],
\end{equation}
with the base case $\mathcal{N}(i, i) = \left\{ (0, 0) \right\}$.

There are two different approaches for polytope representation: (i) vertex representation, which is a set of points, and (ii) half plane representation, which is a set of linear inequalities.
The former is often called $\mathcal{V}$-representation and the latter $\mathcal{H}$-representation. Although they are equivalent, and there are algorithms to transform one into the other, computing Minkowski sum is more convenient with the $\mathcal{V}$-representation, and convex hull of union works more efficiently with the $\mathcal{H}$-representation. The choice of representation and algorithms will affect the running time. In this paper, we use the $\mathcal{V}$-representation.

\subsection{Verification of the Necessary Condition}
Upon computation of $\mathcal{N}(x)$ and $c(x, y)$, the feature vector of the experimentally determined structure, it remains to verify whether $c(x, y) \in \partial \mathcal{N}(x)$. Often, $\mathcal{N}(x)$ is represented by its vertices ($\mathcal{V}$-representation) or its confining half planes ($\mathcal{H}$-representation), two equivalent representations that can be transformed into one another. In an $\mathcal{H}$-representation, $c(x, y)$ is on the boundary of $\mathcal{N}(x)$ iff there is at least one confining plane on which $c(x, y)$ lies. This is true because $c(x, y) \in \mathcal{N}(x)$ anyways. Therefore, the necessary condition can be easily checked by checking membership of $c(x, y)$ in every confining plane. Since the vertices of $\mathcal{N}(x)$ are on the integer lattice, all calculations are rational and hence can be performed exactly.

\subsection{Dataset}
We used 1720 unpseudoknotted RNA sequence-structure pairs from RNA STRAND v2.0 database as our dataset $D$. RNA STRAND v2.0 contains known RNA secondary structures of any type and organism, particularly with and without pseudoknots. To the best of our knowledge, RNA STRAND v2.0 is the most comprehensive collection of known RNA secondary structures to date \cite{Andronescu08}. There are 2334 pseudoknot-free RNAs in the RNA STRAND database. We sorted them based on their length and selected the first 1720 ones, whose lengths vary between 4 and 123 nt. We excluded pseudoknotted structures because our current implementation is incapable of considering pseudoknots. Some sequences in the dataset allow only A-U base pairs (not a single C-G pair), in which case the Newton polytope degenerates into a line.

\subsection{Implementation}
We implemented the dynamic programming in (\ref{equ:dp}) using MATLAB convex hull function which is based on the quickhull algorithm \cite{Barber96}. As mentioned above, we used the $\mathcal{V}$-representation and computed the Minkowski sum by direct pairwise summation of vertices. More precisely, for two convex polytopes $P$ with vertices $p_1, \ldots, p_a$ and $Q$ with vertices $q_1, \ldots, q_b$, the vertices of $P \oplus Q$ are $p_i + q_j$ for $1 \leq i \leq a$ and $1 \leq j \leq b$. To verify the necessary condition, i.e. whether the experimentally determined feature vector lies on the boundary of the Newton polytope, we calculated the distance of the feature vector from the boundary of the polytope using the `p\_poly\_dist' MATLAB function \cite{p-poly-dist}. A zero distance corresponds to the case where the feature vector lies on the boundary, i.e. the condition is satisfied, and a positive distance to the case where the feature vector is in the interior of the Newton polytope. We normalized the distance by square root of the area of the polytope, which is a planar polygon in this case. The normalized distance quantifies how far the feature vector is from the boundary. We parallelized our MATLAB code using MATLAB `parfor'. We ran experiments on the JSLQ nodes of the Wayne State Grid in however non-parallel mode due to lack of support.  The length of input RNA sequences varied between 4 and about 120 nt. For the smallest ones, our program took a fraction of a second and for the longest ones it took less than 10 minutes to run on a 2.4GHz Dual Core AMD Opteron CPU with at most 16 GB RAM.


\section{Results}
For each strand of RNA, the distance between $c(x, y)$, the real feature vector of the secondary structure and the computed convex hull, $\mathcal{N}(x)$, is calculated using \cite{p-poly-dist}. We denote this distance by $r(x)$ here. The necessary condition for the learnability is satisfied if $r(x) = 0$ for all $x$ in the dataset, which shows that the observed feature vector lies on the boundary of $\mathcal{N}(x)$.

\begin{figure}
\begin{center}
\includegraphics[width=0.8\textwidth]{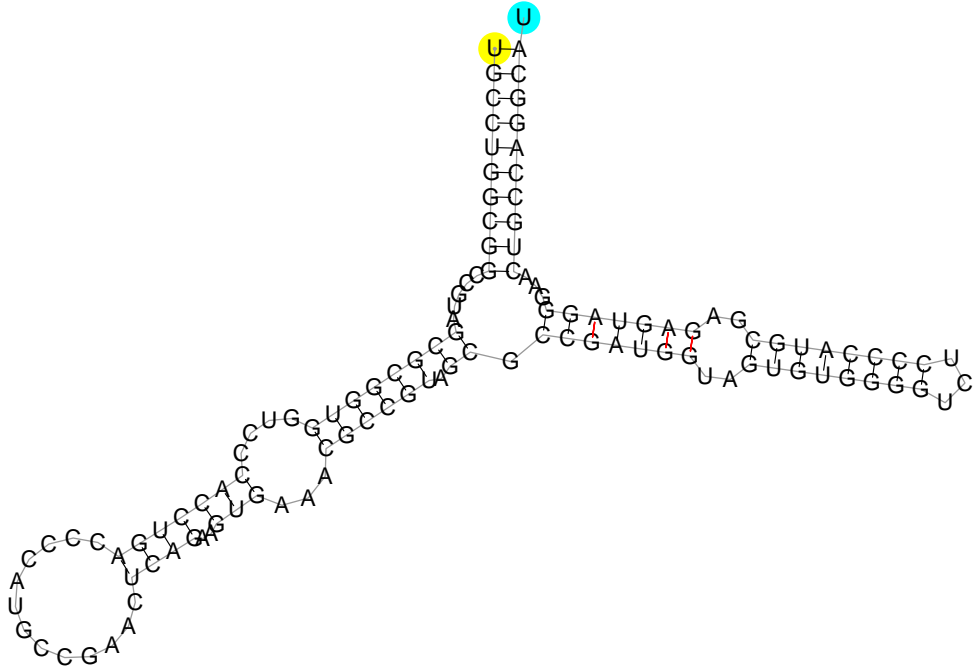} \\
(a) Secondary structure of h.5.b.E.coli.hlxnum. \\
\includegraphics[width=0.8\textwidth]{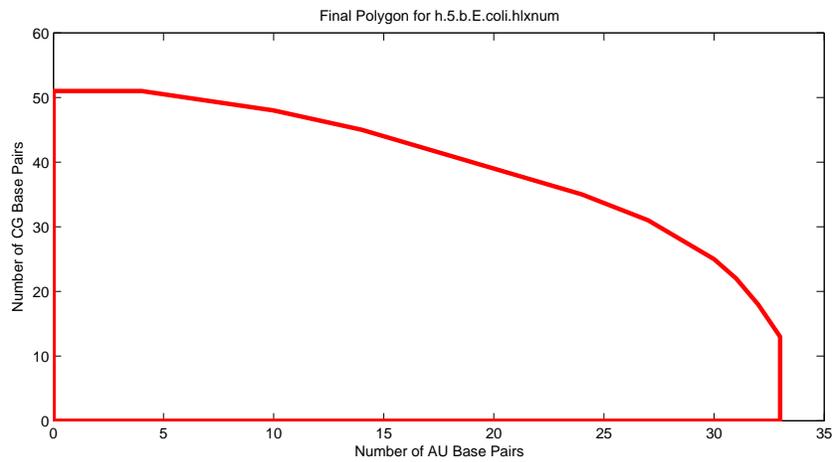}\\
(b) The Newton polygon of h.5.b.E.coli.hlxnum, with 12 vertices.
\end{center}
\caption{The secondary structure and Newton polygon of h.5.b.E.coli.hlxnum in RNA STRAND v2.0 \cite{Andronescu08}. The native feature vector lies on the boundary in this case.} \label{fig:case1}
\end{figure}

\begin{figure}
\begin{center}
\includegraphics[width=0.6\textwidth]{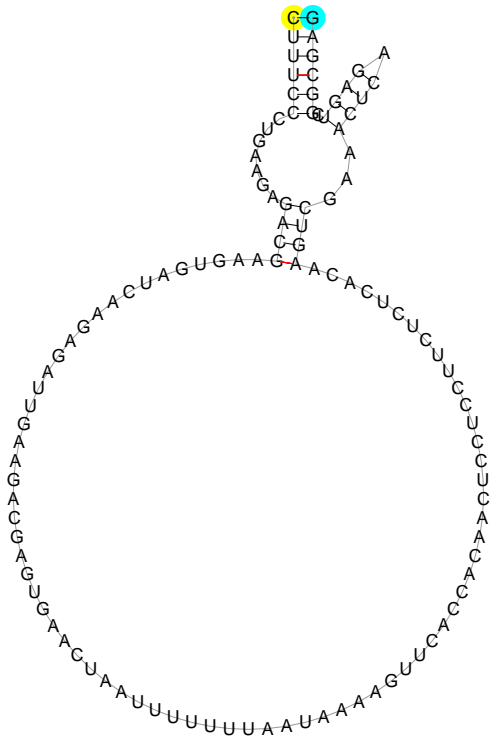} \\
(a) Secondary structure of Hammerhead Ribozyme (Type I). \\\ \\

\includegraphics[width=0.6\textwidth]{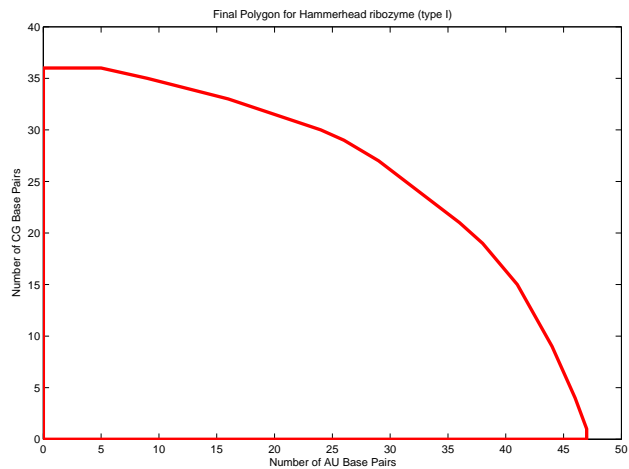} \\
(b) The Newton polygon of Hammerhead Ribozyme (Type I), with 15 vertices.
\end{center}
\caption{The secondary structure and Newton polygon of Hammerhead Ribozyme (Type I) in RNA STRAND v2.0 \cite{Andronescu08}. The native feature vector does not lie on the boundary in this case.} \label{fig:case2}
\end{figure}

\begin{figure}[h!]
\begin{center}
\includegraphics[width=0.8\textwidth]{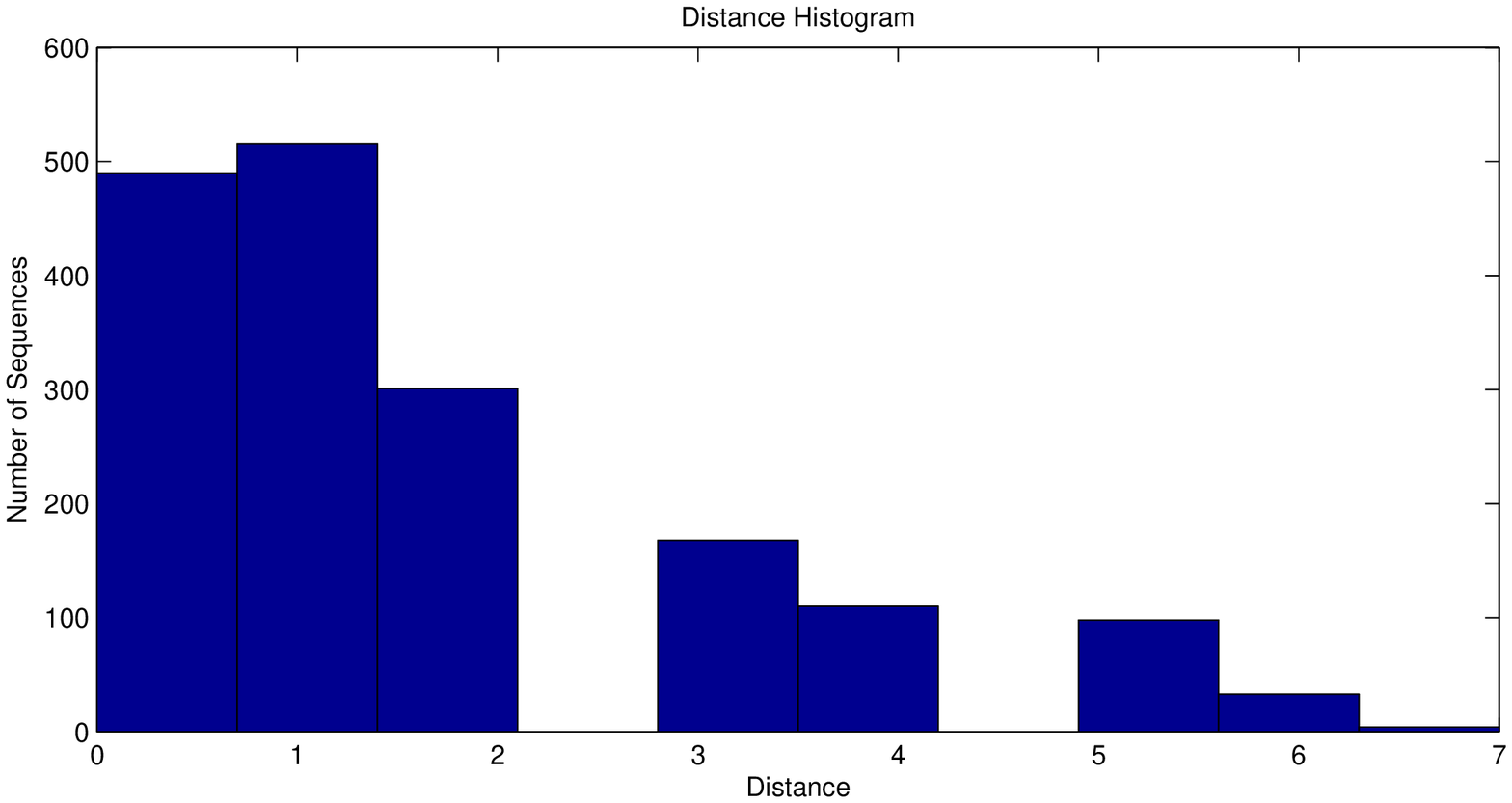}\\
\includegraphics[width=0.8\textwidth]{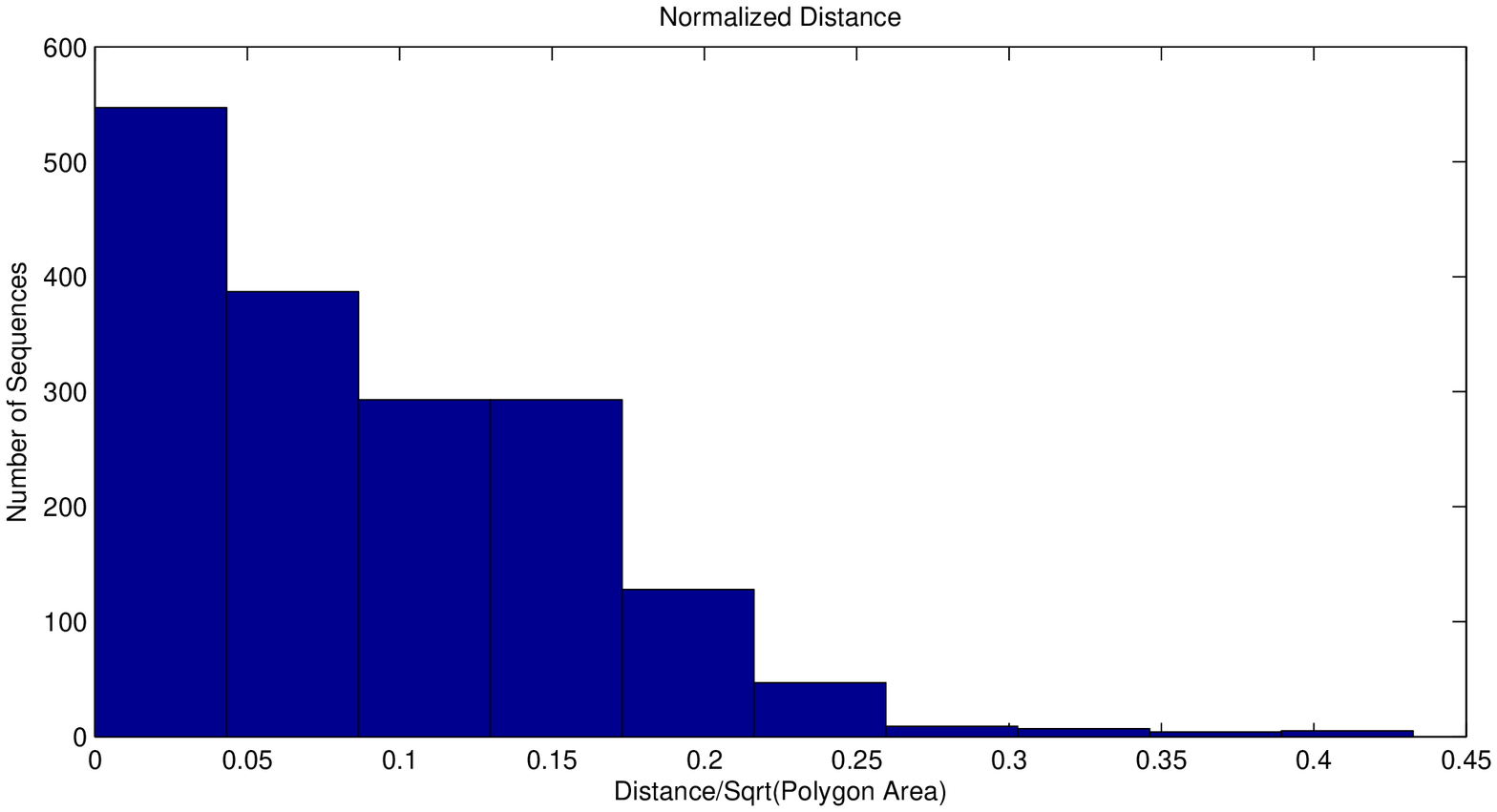}\\
\includegraphics[width=0.8\textwidth]{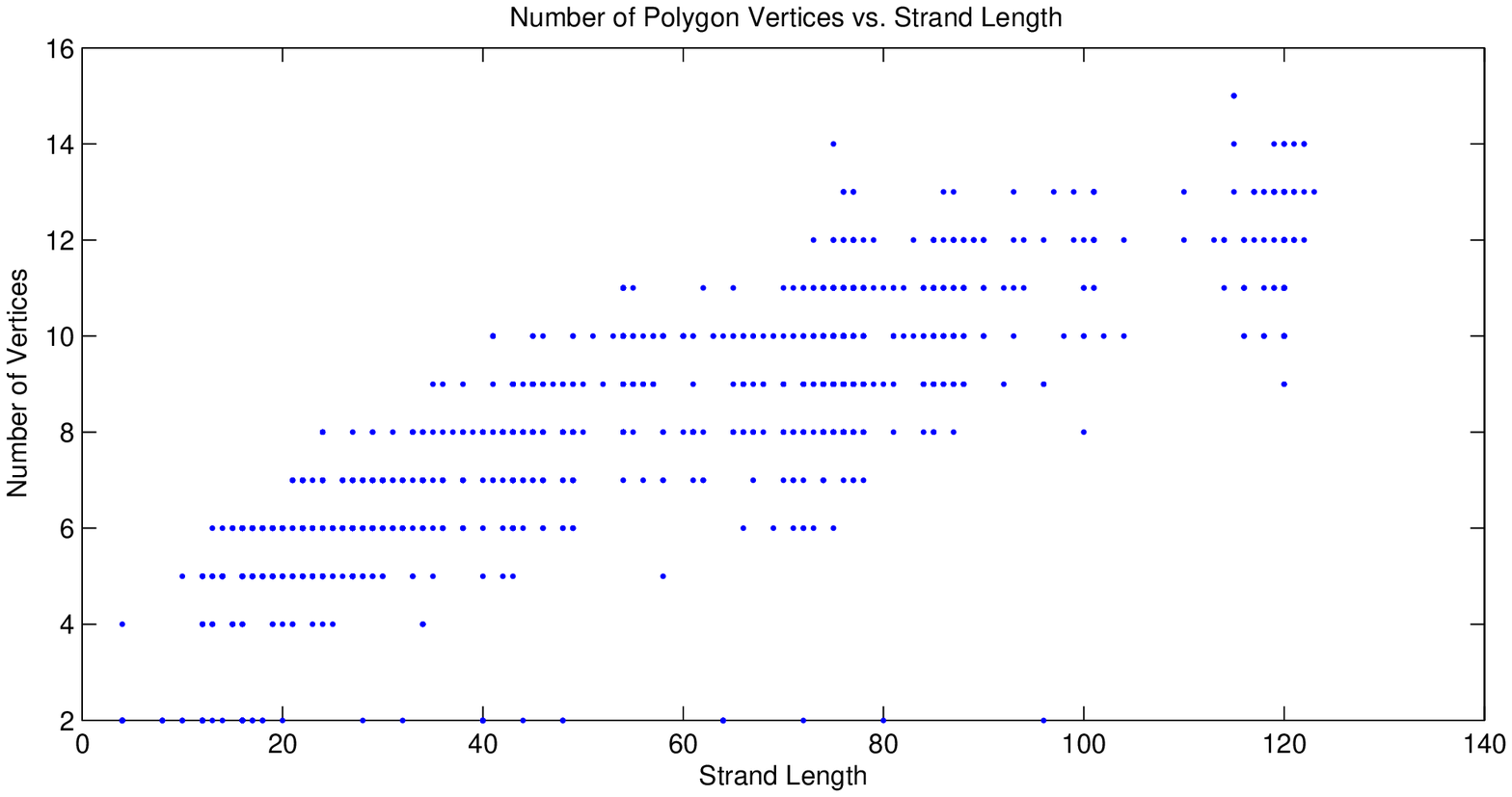}
\end{center}
\caption{(Top) Histogram of $r$. (Middle) Histogram of $r(x) / \sqrt{\mbox{area}(\mathcal{N}(x))}$. (Bottom) Scatter plot of the number of vertices of the Newton polygon as a function of sequence length. }\label{fig:Hists}
\end{figure}

Fig. \ref{fig:case1} illustrates the secondary structure of h.5.b.E.coli.hlxnum, an example in the dataset, and its Newton polygon. This RNA is 120 nt long, and its final Newton polygon has 12 vertices. The native feature vector for this RNA has $r = 0$ which means that it lies on the boundary of the polygon. Fig. \ref{fig:case2} illustrates another example, Hammerhead Ribozyme (Type I) whose length is 115 nt. Its Newton polygon has 15 vertices, and the calculated $r$ is 5 for this case. The big hairpin loop in this case suggests that the dominant energetic features are not base pairs, and this is consistent with the longer distance observed between $\mathcal{N}(x)$ and $c(x, y)$. 

Out of 1720 RNA sequences in our experiment, for 490 (28\%) the native feature vector is exactly on the boundary of the Newton polygon. For the remaining 1230 sequences, the native feature vector settles inside their Newton polygon. Fig. \ref{fig:Hists} demonstrates the histogram of $r(x)$ for the input dataset. For 516 (30\%) of input sequences, $r$ is non-zero but less than 1, which suggests that the dominant energetic features in those cases are canonical base pairs. For 579 (34\%) sequences, $r$ is between 1 and 5. In the remaining 135 cases (8\%), $r$ goes larger than 5. The second plot in Fig. \ref{fig:Hists} shows the normalized distance histogram. The square root of polygon area is used as the normalization factor. For 712 (41\%) strands, this number is no more than 0.05.

The relation between the number of vertices and the length of RNA strand is shown in the third plot of Fig. \ref{fig:Hists}. The maximum number of vertices is 15 which happens for three different strands, each 115 nt long. The minimum number of vertices is two; however, it is clear that no polygon exists with two vertices. In those cases, the RNA admits only A-U base pairs or only C-G base pairs and not both of them. The resulting polygon is just a line, and $c(x, y)$ always lies on the boundary of $\mathcal{N}(x)$. 

\section{Conclusion and Future Work}
We introduced the notion of learnability of the parameters of an energy model as a measure of its inherent capability. We derived a necessary condition for the learnability and gave a dynamic programming algorithm to assess it. Our algorithm computes the convex hull of the feature vectors of all feasible structures in the ensemble of a given input sequence. Also, that convex hull coincides with the \emph{Newton polytope} of the partition function as a polynomial in transformed energy parameters. 

Our theory applied to a simple energy model that counts A-U and C-G base pairs separately revealed that about one third of chosen known structures could potentially be predicted using this simple energy model. For another one third, the necessary condition is barely violated, which suggests that augmenting this energy model with more features is expected to solve the problem. The condition is severely violated for 8\% of sequences, which will be the subject of future investigation. The twilight zone is also interesting and requires deeper examination.   

The Newton polytope lies in the core of computer algebra for solving polynomial equations. Therefore, we envision applications of our RNA Newton polytope in symbolic estimation of energy parameters. Sufficient conditions for the learnability, and also assessing the generalization power of an energy model remain for future work.



\bibliographystyle{unsrt}
\bibliography{pub,ref,main}

\end{document}